\documentclass{article}
\usepackage{ijcai26}

\usepackage{times}
\usepackage{soul}
\usepackage{url}
\usepackage[hidelinks]{hyperref}
\usepackage[utf8]{inputenc}
\usepackage[small]{caption}
\usepackage{graphicx}
\usepackage{amsmath}
\usepackage{amsthm}
\usepackage{booktabs}
\usepackage{algorithm}
\usepackage{algorithmic}
\usepackage[switch]{lineno}

\usepackage{amsmath,amssymb}
\usepackage{mathtools}
\usepackage{bbm}
\usepackage{graphicx}
\usepackage{enumitem}
\usepackage{hyperref}
\usepackage{microtype}
\usepackage{natbib}

\usepackage{xcolor}         % colors
\usepackage{stmaryrd} % Where \llbracket and \rrbracket are defined

\usepackage{graphicx}

\usepackage{amssymb}% http://ctan.org/pkg/amssymb
\usepackage{pifont}% http://ctan.org/pkg/pifont
\newcommand{\cmark}{\ding{51}}%
\newcommand{\xmark}{\ding{55}}%

\usepackage{multirow}

\newtheorem{definition}{Definition}

\newtheorem{theorem}{Theorem}
\newtheorem{proposition}{Proposition}

\title{Differential Voting: Loss Functions For Axiomatically Diverse Aggregation of Heterogeneous Preferences}

\author{
    Zhiyu An
    \qquad
    Duaa Nakshbandi
    \qquad
    Wan Du
    \affiliations
    University of California, Merced
    \emails
    \{zan7, duaanakshbandi, wdu3\}@ucmerced.edu
}

\begin{document}

\maketitle

\begin{abstract}
Reinforcement learning from human feedback (RLHF) implicitly aggregates heterogeneous human preferences into a single utility function, even though the underlying utilities of the participants are in practice diverse. Hence, RLHF can be viewed as a form of voting, where the aggregation mechanism is defined by the loss function. Although Arrow’s Impossibility Theorem suggests that different mechanisms satisfy different sets of desirable axioms, most existing methods rely on a single aggregation principle, typically the Bradley–Terry–Luce (BTL) model, which corresponds to Borda count voting. This restricts the axiomatic properties of the learned reward and obscures the normative assumptions embedded in optimization.
In this work, we introduce Differential Voting, a unifying framework that constructs instance-wise, differentiable loss functions whose population-level optima provably correspond to distinct classical voting rules. We develop differentiable surrogates for majority-based aggregation (BTL), Copeland, and Kemeny rules, and formally analyze their calibration properties, gradient fields, and limiting behavior as smoothing parameters vanish. For each loss, we establish consistency with the corresponding social choice rule and characterize the axioms it satisfies or violates. Our analysis shows how design choices in loss geometry—such as margin sensitivity and boundary concentration—directly translate into normative aggregation behavior.
Differential Voting makes preference aggregation an explicit and controllable design choice in RLHF, enabling principled trade-offs between axiomatic guarantees and optimization stability. Code to reproduce our experiments is open-sourced.\footnote{https://github.com/ryeii/Differential-Voting}

\end{abstract}

\begin{table*}[t]
\centering
\scriptsize
\begin{tabular}{lllccccc}
\toprule
\multirow{2}{*}{Method} &
\multirow{2}{*}{Loss / Objective} &
\multirow{2}{*}{Corresp. Voting System} &
\multicolumn{4}{c}{Axioms} &
\multirow{2}{*}{Diff.} \\
\cmidrule(lr){4-7}
& & &
\begin{tabular}[c]{@{}c@{}}Majority\\ Winner\end{tabular} &
\begin{tabular}[c]{@{}c@{}}Pareto\\ (PO)\end{tabular} &
IIA/LIIA &
\begin{tabular}[c]{@{}c@{}}Condorcet\\ (PMC)\end{tabular} &
\\
\midrule

\textbf{BTL (logistic)} &
$\sum_{a\neq b\in C}\log(1+\exp(r_\theta(b)-r_\theta(a)))$ &
Borda-like scoring &
\xmark & \cmark & \xmark & \xmark & \cmark \\

\textbf{Exponential} &
$\sum_{a\neq b\in C}\exp(r_\theta(b)-r_\theta(a))$ &
Borda-like scoring &
\xmark & \cmark & \xmark & \xmark & \cmark \\

\textbf{Hinge} &
$\sum_{a\neq b\in C}\max(0,1+r_\theta(b)-r_\theta(a))$ &
Borda-like scoring &
\xmark & \cmark & \xmark & \xmark & \cmark \\

\textbf{Majority-based} &
\begin{tabular}[c]{@{}l@{}}
Any rule depending only on\\
pairwise majority directions
\end{tabular} &
(Pairwise) Majority aggregation &
\cmark & \cmark & \xmark & \cmark & \xmark \\

\textbf{Linear Kemeny} &
$\sum_{a\neq b\in C} n_{a\succ b}\cdot \mathbf{1}[r_\theta(b)>r_\theta(a)]$ &
Kemeny (feasible-ranking variant) &
\xmark & \cmark & \xmark & \cmark & \xmark \\

\textbf{LCPO} &
\begin{tabular}[c]{@{}l@{}}
Leximax Copeland subject to PO\\
(implemented via LP feasibility checks)
\end{tabular} &
Copeland (PO-constrained leximax) &
\cmark & \cmark & \xmark & \cmark & \xmark \\

\midrule

\textbf{Soft Copeland (ours)} &
$-y\,s_{\tau,\beta}(\Delta)+\tfrac{\lambda}{2}\Delta^2$ &
Copeland (smooth surrogate) &
\cmark & \cmark & \xmark & \cmark & \cmark \\

\textbf{Soft Kemeny (ours)} &
$\sigma(-y\,\Delta/\tau)$ &
Kemeny (smooth Kendall surrogate) &
\xmark & \cmark & \xmark & \cmark & \cmark \\

\bottomrule
\end{tabular}
\caption{Existing aggregation objectives in (linear) RLHF-style social choice and their axiomatic properties. "Diff." = Differentiable.
Here PO denotes standard Pareto/unanimity (if all voters rank $a$ above $b$, then the social outcome ranks $a$ above $b$),
and PMC denotes the Condorcet winner criterion (if a Condorcet winner exists, it is selected).
Borda-like scoring satisfies PO but violates Majority-winner, IIA/LIIA, and PMC in general; Copeland and Kemeny satisfy PMC but violate IIA/LIIA.
Checkmarks follow the setting and claims discussed in \cite{siththaranjan2023distributional,ge2024axioms} together with standard social-choice facts.}
\label{tab:loss_axioms_completed}
\end{table*}

\section{Introduction}

Reinforcement learning from human feedback (RLHF) has become a dominant paradigm
for aligning machine learning systems with human preferences, particularly in
the training of large language models.
In its standard form, RLHF learns a scalar reward function from pairwise human
comparisons and then optimizes a policy against this learned reward.
Despite its empirical success, this pipeline implicitly aggregates feedback from
multiple annotators, contexts, and latent objectives into a single utility
function, raising a fundamental question:
\emph{what aggregation rule is actually being implemented?}

Recent work has argued that preference learning in RLHF should be understood as a
problem of \emph{social choice}, where heterogeneous human judgments are combined
into a collective decision \cite{ge2024axioms,siththaranjan2023distributional}.
From this perspective, the choice of training loss is not merely a statistical
detail but a normative design decision, encoding assumptions about how conflicting
preferences should be reconciled.
Classical impossibility results, such as Arrow's theorem, imply that no single
aggregation rule can satisfy all desirable axioms simultaneously.
Nevertheless, most existing preference-learning methods rely on a narrow family
of smooth pairwise losses, typically variants of the Bradley--Terry--Luce (BTL)
model, whose induced aggregation behavior is rarely made explicit.

A growing theoretical literature has revealed that this choice has substantive
consequences.
Under heterogeneous or hidden contexts, optimizing BTL-style losses corresponds
to a Borda-like scoring rule at the population level, which can violate basic
axioms such as the Condorcet winner criterion and majority consistency
\cite{siththaranjan2023distributional,ge2024axioms}.
At the same time, alternative aggregation rules with stronger axiomatic
properties—such as Copeland or Kemeny—are typically formulated as discrete or
combinatorial objectives, making them difficult to integrate into gradient-based
training pipelines.

In this work, we bridge this gap by introducing \emph{Differential Voting}, a
unifying framework for designing instance-wise, differentiable loss functions
whose population-level optima correspond to distinct classical voting rules.
Rather than treating preference aggregation as an implicit byproduct of loss
minimization, Differential Voting makes the aggregation mechanism an explicit,
controllable design choice.

We develop differentiable surrogates for three canonical aggregation principles:
pairwise majority (via the standard BTL loss), Copeland aggregation, and Kemeny
aggregation.
For each loss, we characterize its gradient field, prove consistency with the
corresponding social choice rule in the appropriate limit of smoothing
parameters, and analyze which axioms are satisfied or violated.
Our analysis reveals how geometric features of the loss—such as margin
sensitivity, saturation, and boundary concentration—translate directly into
normative aggregation behavior.
In particular, we show that losses that are locally calibrated for pairwise
classification can nonetheless induce globally distinct aggregation rules.

We complement our theoretical results with a suite of controlled synthetic
experiments.
These experiments demonstrate that different losses recover different classical
aggregation rules at the population level, behave differently under
heterogeneous preferences, and exhibit characteristic optimization geometries
that explain their axiomatic properties.
Together, these results show that preference optimization losses are not
interchangeable surrogates but implement qualitatively different social choice
mechanisms.

\paragraph{Contributions.}
Our main contributions are:
\begin{itemize}[leftmargin=1.5em]
    \item We introduce \emph{Differential Voting}, a framework for constructing
    differentiable, instance-wise losses whose population optima correspond to
    classical voting rules.
    \item We propose differentiable surrogates for Copeland and Kemeny aggregation
    and prove their consistency and limiting behavior under smoothing.
    \item We analyze the axiomatic properties and loss geometry of standard and
    proposed objectives, clarifying the normative assumptions encoded by loss
    design.
    \item We empirically validate our theory through controlled experiments that
    isolate aggregation behavior, axiomatic satisfaction, and optimization
    dynamics.
\end{itemize}

\section{Related Work}

\textbf{RLHF and preference optimization.}
Reinforcement learning from human feedback (RLHF) aligns models by learning a
reward function from pairwise human comparisons and optimizing a policy against
it \cite{stiennon2020learning,ouyang2022training}.
Most implementations learn a Bradley--Terry--Luce (BTL) / logistic preference
model and fine-tune with a KL-regularized RL objective.
Direct Preference Optimization (DPO) shows that an equivalent objective can be
optimized via a classification-style loss under a particular reward
parameterization \cite{rafailov2023direct}.
These methods build on classical pairwise ranking losses such as RankNet, which
optimize smooth surrogates of ranking quality using logistic link functions
\cite{burges2005learning}.

\textbf{Social choice and axioms for alignment.}
Recent work reframes preference learning in RLHF as a problem of \emph{preference
aggregation}, emphasizing evaluation through axioms from social choice theory.
Ge et al.\ introduce an axiomatic framework for learning reward functions from
comparisons and show that BTL and broad convex generalizations violate basic
desiderata in their \emph{linear social choice} setting \cite{ge2024axioms}.
They further demonstrate that function-approximation constraints fundamentally
affect which axioms are achievable, motivating aggregation rules such as LCPO
that satisfy stronger guarantees than standard loss-based methods.

\textbf{Heterogeneous preferences and hidden context.}
Feedback in deployed alignment is inherently heterogeneous, reflecting diverse
annotators and latent objectives.
Siththaranjan et al.\ formalize this via \emph{hidden context} and show that
standard preference learning implicitly aggregates across contexts according to
a Borda-like scoring rule, which can lead to counterintuitive outcomes
\cite{siththaranjan2023distributional}.
This work highlights that the choice of loss encodes normative assumptions about
how conflicting judgments are combined.

\textbf{Pluralism and alternative aggregation targets.}
Complementary lines of work study alignment under pluralistic or adversarial
settings.
G{\"o}lz et al.\ analyze \emph{distortion} between achieved and optimal welfare,
sharply distinguishing RLHF/DPO from alternatives such as Nash Learning from
Human Feedback (NLHF) \cite{golz2025distortion}.
Other approaches modify the aggregation target itself: maximal lotteries provide
a probabilistic Condorcet extension for cyclic preferences
\cite{maura2025jackpot}, while NLHF frames alignment as computing equilibria of
learned preference models \cite{munos2024nash}.

\textbf{Positioning.}
Prior work shows that standard differentiable losses implicitly implement
specific aggregation rules that can violate appealing axioms, while methods with
stronger normative guarantees often rely on discrete or non-smooth objectives.
Differential Voting complements these approaches by making the aggregation rule
an explicit \emph{loss-level} design choice, enabling gradient-based optimization
while targeting distinct voting rules via calibrated, differentiable
surrogates.

\section{Differential Voting: From Classical Aggregation to Differentiable Losses}
\label{sec:differential-voting}

We view RLHF as \emph{implicit preference aggregation}: heterogeneous pairwise
judgments are combined into a single scalar reward function, and the choice of
loss determines which aggregation principle is implemented at the population
optimum. We construct instance-wise, differentiable losses whose induced
population objectives correspond to classical voting rules.

Let $\mathcal{D}$ be a distribution over pairwise comparisons $(x_i,x_j,y,c)$,
where $y\in\{-1,+1\}$ indicates whether $x_i \succ x_j$ under context $c$.
A reward model $r_\theta(x,c)\in\mathbb{R}$ induces the margin
$\Delta_{ij}(c)=r_\theta(x_i,c)-r_\theta(x_j,c)$ and win probability
\begin{equation}
\label{eq:winprob}
p_\theta(x_i \succ x_j\mid c)=\sigma\!\left(\frac{\Delta_{ij}(c)}{\tau}\right),
\qquad
\sigma(z)=\frac{1}{1+e^{-z}},
\end{equation}
with temperature $\tau>0$. Define the true pairwise preference probability
\[
\eta(x_i,x_j,c) := \Pr(y=+1\mid x_i,x_j,c).
\]
We also use the shorthand $\Delta=\Delta_{ij}(c)$ and $p=\sigma(\Delta/\tau)$.

Given a context $c$ (or suppressing $c$ when fixed), the learned reward induces
an ordering by sorting $r_\theta(\cdot,c)$. Our results concern the population
optima of the losses and their relationship to the underlying pairwise
probabilities $\eta$ and classical social choice rules.

\paragraph{Population optima, realizability, and scope of claims.}
All consistency and axiomatic statements in this section concern \emph{population}
objectives and their minimizers over the hypothesis class
$\{r_\theta(\cdot,c):\theta\in\Theta\}$.
When we say that a loss ``implements'' or ``corresponds to'' a classical voting
rule, this means that the induced population objective coincides with the
objective of that rule \emph{restricted to the set of rankings realizable by the
model class}.
If the hypothesis class is rich enough to represent a ranking selected by the
classical rule, then the corresponding population minimizer coincides with that
rule; otherwise, the minimizer coincides with the rule \emph{restricted to the
feasible set}.
We do not make claims about optimization dynamics or convergence of specific
algorithms, only about population-level optima.

\subsection{BTL / Logistic Loss: Local Majority Calibration}
\label{subsec:btl-fisher}

We begin with the standard Bradley--Terry--Luce (BTL) / logistic loss as a
baseline. Prior work shows that, when pairwise losses are aggregated across
alternatives and heterogeneous contexts, minimizing the BTL objective induces a
Borda-like scoring rule at the population level
\cite{siththaranjan2023distributional,ge2024axioms}.
Here we record a complementary and well-known fact at a different level of
analysis: \emph{for a fixed pair under a fixed context}, the BTL loss is
classification-calibrated for the pairwise majority direction.

This result is classical, but we state it explicitly to clarify the distinction
between \emph{local statistical calibration} and \emph{global aggregation
behavior}, a theme that will recur for Copeland- and Kemeny-style losses.

\paragraph{Setup.}
Fix a comparison $(x_i,x_j,c)$ and let
\[
\eta := \Pr(y=+1 \mid x_i,x_j,c).
\]
Let $\Delta = r_\theta(x_i,c)-r_\theta(x_j,c)$.

\begin{definition}[BTL / logistic loss]
\label{def:btl}
\[
\mathcal{L}_{\mathrm{BTL}}(\Delta,y)
=
\log\!\left(1+\exp\!\left(-\frac{y\Delta}{\tau}\right)\right),
\qquad \tau>0 .
\]
\end{definition}

\begin{proposition}[Pairwise majority calibration]
\label{prop:btl-fisher}
Consider the conditional risk
\begin{multline*}
    \mathcal{R}_{\mathrm{BTL}}(\Delta)
=
\mathbb{E}\big[\mathcal{L}_{\mathrm{BTL}}(\Delta,Y)\mid x_i,x_j,c\big]\\
=
\eta\log(1+e^{-\Delta/\tau})+(1-\eta)\log(1+e^{\Delta/\tau}).
\end{multline*}
\begin{itemize}[leftmargin=1.5em]
\item If $\eta\in(0,1)$, $\mathcal{R}_{\mathrm{BTL}}$ is strictly convex and has the
unique minimizer
\[
\Delta^\star = \tau\log\frac{\eta}{1-\eta},
\qquad
\operatorname{sign}(\Delta^\star)=\operatorname{sign}\!\big(\eta-\tfrac12\big).
\]
\item If $\eta=1$ (resp.\ $\eta=0$), the risk has no finite minimizer and its
infimum is approached as $\Delta\to+\infty$ (resp.\ $\Delta\to-\infty$).
\end{itemize}
Thus, whenever $\eta\neq\tfrac12$, minimizing the BTL loss recovers the correct
pairwise majority direction.
\end{proposition}

\paragraph{Remark (local vs.\ global behavior).}
Proposition~\ref{prop:btl-fisher} is a \emph{local} statement: it concerns the
optimal margin for a single ordered pair $(x_i,x_j)$ under a fixed context.
It does not characterize the \emph{multi-alternative aggregation rule} induced
by minimizing sums of pairwise BTL losses across datasets.
As shown in prior work, when such losses are aggregated under heterogeneous
contexts and standard sampling schemes, the resulting population objective
implements a Borda-like scoring principle rather than a majority rule
\cite{siththaranjan2023distributional,ge2024axioms}.

% ============================================================
\subsection{Soft Copeland Voting}
\label{subsec:cop}

Copeland aggregation counts only whether an alternative wins or loses each
head-to-head contest, ignoring margin magnitudes. A differentiable surrogate
should therefore (i) be odd in the margin, (ii) saturate for confident
wins/losses, and (iii) concentrate gradient mass near the win/loss boundary.
At the same time, to obtain a well-posed population objective with finite
minimizers, some form of regularization is required.

\paragraph{Soft Copeland edge score.}
Define the soft Copeland edge score
\begin{equation}
\label{eq:cop-score}
s_{\tau,\beta}(\Delta)
=
\tanh\!\left[\beta\left(\sigma(\Delta/\tau)-\tfrac12\right)\right],
\end{equation}
where $\sigma(z)=(1+e^{-z})^{-1}$, $\tau>0$ controls smoothness of the pairwise
decision boundary, and $\beta>0$ controls saturation.
The function $s_{\tau,\beta}$ is odd, bounded in $[-1,1]$, and strictly increasing
in $\Delta$.

\begin{definition}[Regularized Soft Copeland Loss]
\label{def:cop}
For one instance $(x_i,x_j,y,c)$ define
\begin{equation}
\label{eq:cop-loss}
\mathcal{L}_{\mathrm{Cop}}(\theta;x_i,x_j,y,c)
=
-\,y\,s_{\tau,\beta}(\Delta_{ij}(c))
\;+\;
\frac{\lambda}{2}\,\Delta_{ij}(c)^2,
\end{equation}
where $\lambda>0$ is a (small) margin-regularization parameter.
\end{definition}

\paragraph{Remark (well-posedness).}
Without the quadratic term, the conditional risk is minimized only by sending
$\Delta\to\pm\infty$ whenever $\eta\neq\tfrac12$.
The regularizer makes the objective coercive and yields finite population optima
while preserving the Copeland-style saturation geometry.
In practice, such regularization corresponds to explicit weight decay, KL
penalties, or implicit regularization from early stopping.

\paragraph{Gradient field.}
Let $p=\sigma(\Delta/\tau)$.
Since
\[
\frac{d}{d\Delta}s_{\tau,\beta}(\Delta)
=
\frac{\beta}{\tau}\,p(1-p)\,\mathrm{sech}^2\!\big(\beta(p-\tfrac12)\big),
\]
we obtain
\begin{equation}
\label{eq:cop-dLdDelta}
\frac{\partial \mathcal{L}_{\mathrm{Cop}}}{\partial \Delta}
=
-\,y\;\frac{\beta}{\tau}\;p(1-p)\;\mathrm{sech}^2\!\big(\beta(p-\tfrac12)\big)
\;+\;\lambda\Delta.
\end{equation}
Parameter gradients follow by multiplying by
$\nabla_\theta\Delta_{ij}(c)=\nabla_\theta r_\theta(x_i,c)-\nabla_\theta r_\theta(x_j,c)$.

\paragraph{Pairwise Copeland consistency.}
We characterize the population-optimal margin for a fixed pair.

\begin{theorem}[Pairwise Copeland direction consistency with finite optima]
\label{thm:cop-pairwise-consistency}
Fix $(x_i,x_j,c)$ and let $\eta=\Pr(y=+1\mid x_i,x_j,c)$.
The conditional risk
\[
\mathcal{R}_{\mathrm{Cop}}(\Delta)
=
-(2\eta-1)\,s_{\tau,\beta}(\Delta)
+\frac{\lambda}{2}\Delta^2
\]
admits at least one finite minimizer $\Delta^\star$.
Moreover,
\[
\operatorname{sign}(\Delta^\star)=\operatorname{sign}(2\eta-1)
\qquad\text{whenever }\eta\neq\tfrac12,
\]
and if $\eta=\tfrac12$ then $\Delta^\star=0$ is the unique minimizer.
Finally, as $\lambda\downarrow 0$ with $(\tau,\beta)$ fixed, any sequence of
minimizers satisfies $|\Delta^\star|\to\infty$ when $\eta\neq\tfrac12$.
\end{theorem}

\begin{proof}
Since $s_{\tau,\beta}(\Delta)\in[-1,1]$ and $\frac{\lambda}{2}\Delta^2\to\infty$
as $|\Delta|\to\infty$, $\mathcal{R}_{\mathrm{Cop}}$ is coercive and attains a
finite minimum.

Write $\gamma:=2\eta-1$. The function $s_{\tau,\beta}$ is odd and strictly
increasing, while $\Delta^2$ is even. Hence for any $\Delta>0$,
\begin{multline*}
    \mathcal{R}_{\mathrm{Cop}}(\Delta)-\mathcal{R}_{\mathrm{Cop}}(-\Delta)
=
-\gamma\bigl(s_{\tau,\beta}(\Delta)-s_{\tau,\beta}(-\Delta)\bigr)\\
=
-2\gamma\,s_{\tau,\beta}(\Delta).
\end{multline*}
If $\eta>\tfrac12$ (so $\gamma>0$), then $s_{\tau,\beta}(\Delta)>0$ for $\Delta>0$,
and therefore $\mathcal{R}_{\mathrm{Cop}}(\Delta)<\mathcal{R}_{\mathrm{Cop}}(-\Delta)$.
Consequently, no minimizer can be negative; by symmetry, if $\eta<\tfrac12$ then
no minimizer can be positive.

It remains to rule out $\Delta^\star=0$ when $\eta\neq\tfrac12$.
Differentiate:
\[
\mathcal{R}'_{\mathrm{Cop}}(\Delta)= -\gamma\,s'_{\tau,\beta}(\Delta)+\lambda\Delta.
\]
Since $s_{\tau,\beta}$ is strictly increasing and smooth, $s'_{\tau,\beta}(0)>0$,
and thus
\[
\mathcal{R}'_{\mathrm{Cop}}(0)= -\gamma\,s'_{\tau,\beta}(0).
\]
If $\eta>\tfrac12$ then $\gamma>0$ and $\mathcal{R}'_{\mathrm{Cop}}(0)<0$, so $0$
cannot be a minimizer; hence every minimizer satisfies $\Delta^\star>0$.
Similarly, if $\eta<\tfrac12$ then $\mathcal{R}'_{\mathrm{Cop}}(0)>0$ and every
minimizer satisfies $\Delta^\star<0$.
If $\eta=\tfrac12$ then $\gamma=0$ and the risk reduces to $\frac{\lambda}{2}\Delta^2$,
uniquely minimized at $\Delta=0$.

Finally, when $\eta\neq\tfrac12$ and $\lambda\downarrow 0$ with $(\tau,\beta)$ fixed,
the bounded term $-\gamma s_{\tau,\beta}(\Delta)$ is minimized by pushing
$s_{\tau,\beta}(\Delta)$ toward its extremum, which occurs only as $|\Delta|\to\infty$.
\end{proof}

\paragraph{From pairwise direction to Copeland score.}
Define the induced soft Copeland score of an alternative $x$ by
\begin{equation}
\label{eq:cop-score-pop}
C_{\tau,\beta}(x)
=
\mathbb{E}_{x'\sim \rho}\Big[s_{\tau,\beta}\big(r_\theta(x,c)-r_\theta(x',c)\big)\Big],
\end{equation}
where $\rho$ is the comparison-opponent distribution induced by $\mathcal{D}$.

\begin{theorem}[Limit to classical Copeland counting]
\label{thm:cop-limit}
Assume $r_\theta(\cdot,c)$ induces strict inequalities almost surely.
Then for each fixed $\theta$,
\[
\lim_{\tau\to 0}\lim_{\beta\to\infty} C_{\tau,\beta}(x)
=
\mathbb{E}_{x'\sim\rho}\Big[\operatorname{sign}\big(r_\theta(x,c)-r_\theta(x',c)\big)\Big],
\]
which is exactly the (normalized) Copeland win-minus-loss count under the induced
pairwise outcomes.
\end{theorem}

\begin{proof}
Fix $\tau>0$ and write $\Delta=r_\theta(x,c)-r_\theta(x',c)$.
For $\Delta\neq 0$, we have $\sigma(\Delta/\tau)-\tfrac12$ has the same sign as $\Delta$,
and hence as $\beta\to\infty$,
\begin{multline*}
    s_{\tau,\beta}(\Delta)
=
\tanh\!\left(\beta\left(\sigma(\Delta/\tau)-\tfrac12\right)\right)\\
\longrightarrow
\operatorname{sign}\!\left(\sigma(\Delta/\tau)-\tfrac12\right)
=
\operatorname{sign}(\Delta).
\end{multline*}
Because $|s_{\tau,\beta}(\Delta)|\le 1$, dominated convergence yields
\[
\lim_{\beta\to\infty} C_{\tau,\beta}(x)
=
\mathbb{E}_{x'\sim\rho}\big[\operatorname{sign}(\Delta)\big].
\]
The right-hand side does not depend on $\tau$, so taking $\tau\to 0$ leaves it unchanged,
proving the claim.
\end{proof}

\paragraph{Axiomatic properties.}
\begin{proposition}[Neutrality and anonymity]
\label{prop:cop-neutral-anon}
The Soft Copeland loss is neutral and anonymous.
\end{proposition}

\begin{proof}
Neutrality follows from the oddness of $s_{\tau,\beta}$ and the invariance of
$\Delta^2$ under sign flips.
Anonymity holds because empirical and population risks are averages of identical
per-instance losses.
\end{proof}

\begin{proposition}[Pairwise monotonicity]
\label{prop:cop-monotone}
For any instance with $y=+1$ and any $\Delta$ satisfying
$|\Delta| < \frac{1}{\lambda}\frac{\beta}{4\tau}$,
we have $\frac{\partial \mathcal{L}_{\mathrm{Cop}}}{\partial \Delta}<0$.
The symmetric statement holds for $y=-1$.
\end{proposition}

\begin{proof}
From \eqref{eq:cop-dLdDelta}, the first term has sign $-y$ and magnitude at most
$\frac{\beta}{4\tau}$, while the regularization term has magnitude $\lambda|\Delta|$.
For $|\Delta|<\frac{1}{\lambda}\frac{\beta}{4\tau}$, the former dominates, yielding
the stated sign.
\end{proof}

\paragraph{IIA failure.}
As with classical Copeland, the induced aggregation violates independence of
irrelevant alternatives.
Since Soft Copeland converges to Copeland counting in the limit
(Theorem~\ref{thm:cop-limit}), it inherits the standard IIA counterexamples.

% ============================================================
\subsection{Soft Kemeny Voting}
\label{subsec:kem}

Kemeny aggregation selects a ranking that minimizes the total number of pairwise
disagreements (equivalently, the Kendall $\tau$ distance) with respect to observed
comparisons.
To approximate this objective using gradient-based optimization, we use a smooth
surrogate that directly targets pairwise disagreement and converges to exact
disagreement counting as the temperature vanishes.

\paragraph{Smooth disagreement surrogate.}
For a comparison $(x_i,x_j,y,c)$ with margin
\[
\Delta=\Delta_{ij}(c)=r_\theta(x_i,c)-r_\theta(x_j,c),
\]
define the \emph{Soft Kemeny loss}
\begin{equation}
\label{eq:kem-loss}
\mathcal{L}_{\mathrm{Kem}}(\theta;x_i,x_j,y,c)
=
\sigma\!\left(-\frac{y\Delta_{ij}(c)}{\tau}\right),
\qquad \tau>0,
\end{equation}
where $\sigma(z)=(1+e^{-z})^{-1}$.

This loss is a smooth approximation to the pairwise disagreement indicator:
for $\Delta\neq 0$,
\[
\sigma\!\left(-\frac{y\Delta}{\tau}\right)
\;\longrightarrow\;
\mathbf{1}[y\Delta<0]
\qquad (\tau\to 0).
\]

\paragraph{Gradient field.}
Writing $\Delta=\Delta_{ij}(c)$, we have
\begin{equation}
\label{eq:kem-dLdDelta}
\frac{\partial \mathcal{L}_{\mathrm{Kem}}}{\partial \Delta}
=
-\frac{y}{\tau}\,
\sigma\!\left(-\frac{y\Delta}{\tau}\right)
\left(1-\sigma\!\left(-\frac{y\Delta}{\tau}\right)\right).
\end{equation}
Parameter gradients follow by multiplying
$\nabla_\theta\Delta_{ij}(c)
=
\nabla_\theta r_\theta(x_i,c)-\nabla_\theta r_\theta(x_j,c)$.

\paragraph{Global pairwise monotonicity.}
The loss always corrects pairwise errors in the appropriate direction.

\begin{proposition}[Global pairwise monotonicity]
\label{prop:kem-monotone}
Fix an instance with label $y\in\{-1,+1\}$.
Then for all $\Delta\in\mathbb{R}$,
\[
\operatorname{sign}\!\left(
\frac{\partial \mathcal{L}_{\mathrm{Kem}}}{\partial \Delta}
\right)
=
-\,y .
\]
Thus, gradient descent increases $\Delta$ when $y=+1$ and decreases $\Delta$ when
$y=-1$, for all margin values.
\end{proposition}

\begin{proof}
From \eqref{eq:kem-dLdDelta}, the derivative is proportional to $-y$ multiplied by
$\sigma(-y\Delta/\tau)(1-\sigma(-y\Delta/\tau))\ge 0$, yielding the stated sign.
\end{proof}

\paragraph{Vanishing gradients for confident correct orderings.}
If $y\Delta\to +\infty$, then
$\sigma(-y\Delta/\tau)\to 0$ and
\[
\left|
\frac{\partial \mathcal{L}_{\mathrm{Kem}}}{\partial \Delta}
\right|
\;\longrightarrow\;0 .
\]
Thus, correctly ordered pairs with large margins contribute vanishing gradient
mass, matching the intuition of disagreement counting: once a pair is confidently
ordered correctly, there is no incentive to further increase its margin.

\paragraph{Limit to Kendall disagreement.}
We recover exact pairwise disagreement counting in the low-temperature limit.

\begin{theorem}[Limit to pairwise disagreement objective]
\label{thm:kem-limit-objective}
Assume $\Pr(\Delta_{ij}(c)=0)=0$ under $\mathcal{D}$.
Then for each fixed $\theta$,
\[
\lim_{\tau\to 0}
\mathbb{E}_{\mathcal{D}}\big[\mathcal{L}_{\mathrm{Kem}}(\theta)\big]
=
\mathbb{E}_{\mathcal{D}}\big[\mathbf{1}[y\Delta_{ij}(c)<0]\big].
\]
\end{theorem}

\begin{proof}
Fix $\Delta\neq 0$ and $y\in\{\pm 1\}$.
If $y\Delta<0$, then $-y\Delta/\tau\to +\infty$ as $\tau\to 0$, so
$\sigma(-y\Delta/\tau)\to 1$.
If $y\Delta>0$, then $-y\Delta/\tau\to -\infty$, so
$\sigma(-y\Delta/\tau)\to 0$.
Thus $\mathcal{L}_{\mathrm{Kem}}(\theta)\to \mathbf{1}[y\Delta<0]$ pointwise.
Since $0\le \mathcal{L}_{\mathrm{Kem}}\le 1$, dominated convergence applies.
\end{proof}

\begin{figure*}[t]
    \centering
    \includegraphics[width=\linewidth]{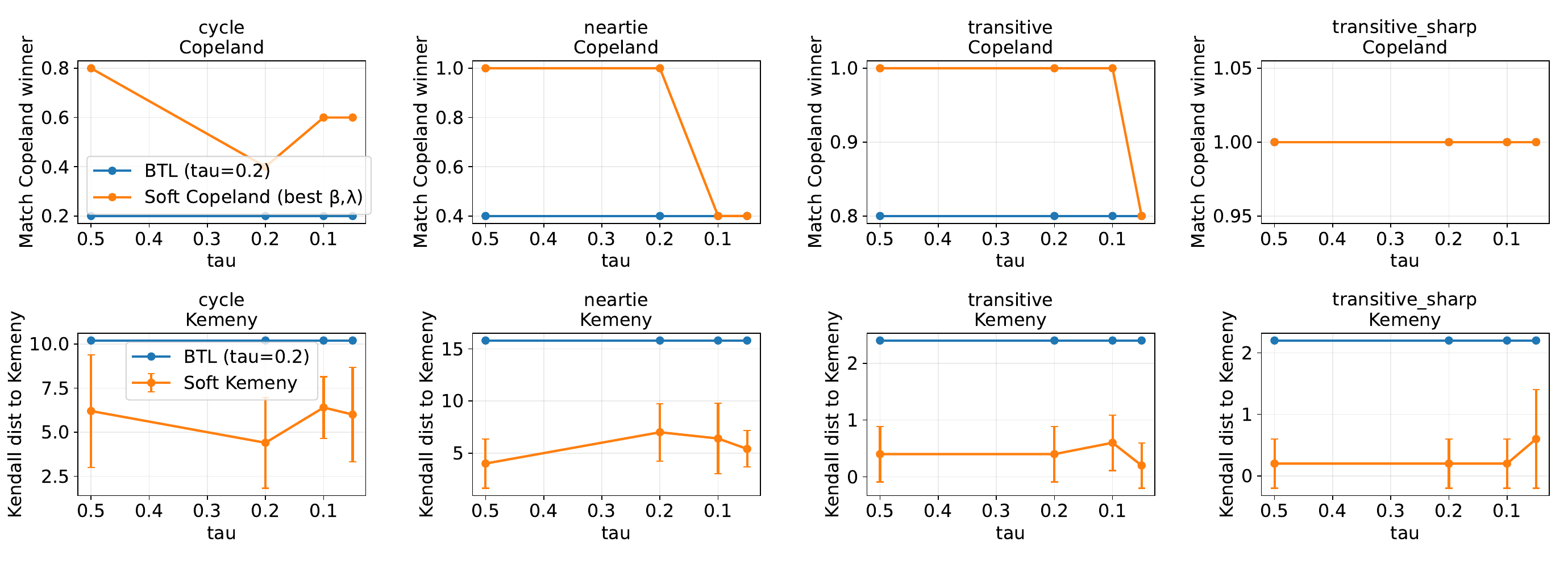}
    \vspace{-2em}
    \caption{Population-level recovery of classical aggregation rules under smoothing.
    Top row: agreement with the Copeland winner; bottom row: Kendall distance to the Kemeny-optimal ranking.
    Soft Copeland and Soft Kemeny converge to their respective classical rules with appropriate $\tau$,
    while BTL remains misaligned across cyclic, near-tie, and transitive regimes.}
    \label{fig:exp1}
\end{figure*}

\begin{theorem}[Kemeny consistency under realizability]
\label{thm:kem-consistency}
Let $\pi_\theta$ be the strict ranking induced by sorting $r_\theta(\cdot,c)$ and
$\Pi_\Theta=\{\pi_\theta:\theta\in\Theta\}$.
Assume the conditions of Theorem~\ref{thm:kem-limit-objective} hold.

For any sequence $\tau_n \downarrow 0$, let
\[
\theta_n \in \arg\min_{\theta\in\Theta}
\mathbb{E}\big[\mathcal{L}_{\mathrm{Kem}}^{(\tau_n)}(\theta)\big].
\]
Then every accumulation point $\theta^\star$ of $\{\theta_n\}$ induces a ranking
$\pi_{\theta^\star}$ satisfying
\[
\pi_{\theta^\star}
\;\in\;
\arg\min_{\pi\in\Pi_\Theta}
\mathbb{E}\big[\mathbf{1}[\pi \text{ disagrees with } y]\big].
\]

In particular, if a Kemeny-optimal ranking is realizable by $\Pi_\Theta$, then any
such limit point induces a Kemeny-optimal ranking; otherwise, it induces a
Kemeny-optimal ranking restricted to the feasible set $\Pi_\Theta$.
\end{theorem}

\begin{proof}
For a comparison $(x_i,x_j,y,c)\sim\mathcal{D}$ write
$\Delta_\theta=r_\theta(x_i,c)-r_\theta(x_j,c)$ and
$T_\theta=y\,\Delta_\theta$.
For $\tau>0$ define the population risks
\[
F_\tau(\theta)=\mathbb{E}\!\left[\sigma\!\left(-\frac{T_\theta}{\tau}\right)\right],
\qquad
F_0(\theta)=\mathbb{E}\big[\mathbf{1}[T_\theta<0]\big].
\]
By Theorem~\ref{thm:kem-limit-objective}, $F_\tau(\theta)\to F_0(\theta)$ pointwise
for every fixed $\theta$.

Let $\tau_n\downarrow 0$ and choose
$\theta_n\in\arg\min_{\theta\in\Theta}F_{\tau_n}(\theta)$.
Let $\theta^\star$ be an accumulation point of $\{\theta_n\}$ (passing to a
subsequence if necessary).

We first show that
\begin{equation}
\label{eq:kem-risk-convergence}
\lim_{n\to\infty}F_{\tau_n}(\theta_n)=F_0(\theta^\star).
\end{equation}
Since $\theta_n\to\theta^\star$, we have $T_{\theta_n}\to T_{\theta^\star}$ almost
surely.
Under the assumption $\Pr(T_{\theta^\star}=0)=0$, the sign of $T_{\theta_n}$
eventually agrees with that of $T_{\theta^\star}$.
Hence
\[
\sigma\!\left(-\frac{T_{\theta_n}}{\tau_n}\right)
\;\longrightarrow\;
\mathbf{1}[T_{\theta^\star}<0]
\qquad\text{almost surely}.
\]
Because the integrand is bounded in $[0,1]$, dominated convergence yields
\eqref{eq:kem-risk-convergence}.

Now fix any $\theta\in\Theta$.
By optimality of $\theta_n$ for $F_{\tau_n}$,
\[
F_{\tau_n}(\theta_n)\le F_{\tau_n}(\theta)\qquad\text{for all }n.
\]
Taking limits, using \eqref{eq:kem-risk-convergence} on the left and
Theorem~\ref{thm:kem-limit-objective} on the right, gives
\[
F_0(\theta^\star)\le F_0(\theta).
\]
Since $\theta$ was arbitrary, $\theta^\star\in\arg\min_{\theta\in\Theta}F_0(\theta)$.

Finally, $F_0(\theta)$ depends on $\theta$ only through the strict ranking
$\pi_\theta$ induced by $r_\theta(\cdot,c)$.
Thus minimizing $F_0$ over $\theta\in\Theta$ is equivalent to minimizing expected
pairwise disagreement over realizable rankings $\Pi_\Theta$.
Consequently, $\pi_{\theta^\star}$ is Kemeny-optimal within $\Pi_\Theta$.
If a classical Kemeny-optimal ranking is realizable, it is recovered; otherwise
the optimum is attained within the feasible set.
\end{proof}

\paragraph{Axiomatic properties.}
Neutrality and anonymity follow from invariance of the pairwise margins
$\Delta_{ij}(c)$ under candidate relabeling and from the instance-wise additive
structure of the empirical and population objectives.

Soft Kemeny directly targets pairwise disagreement rather than margin
maximization.
The logistic form concentrates gradient mass near decision boundaries, yields
vanishing gradients for confidently correct orderings, and converges exactly to
Kendall disagreement counting as the temperature vanishes.

% ============================================================

% \begin{table*}[h!]
% \scriptsize
% \centering
% \caption{Differential voting losses and their classical interpretations.
% Each loss is instance-wise and differentiable.
% ``Local'' refers to Fisher-style calibration for a fixed pair under a fixed context,
% while ``Global'' refers to the induced population objective over datasets of pairwise comparisons,
% restricted to realizable rankings.}
% \begin{tabular}{lccc}
% \toprule
%  & \textbf{BTL (Logistic)} & \textbf{Soft Copeland} & \textbf{Soft Kemeny} \\
% \midrule

% \textbf{Classical target}
% & Pairwise majority (local)
% & Copeland
% & Kemeny \\[0.4em]

% \textbf{Per-instance loss}
% & $\log\!\bigl(1+\exp(-y\,\Delta/\tau)\bigr)$
% & $-y\,\tanh\!\Bigl[\beta\bigl(\sigma(\Delta/\tau)-\tfrac12\bigr)\Bigr]
% +\tfrac{\lambda}{2}\Delta^2$
% & $\sigma\!\bigl(-y\,\Delta/\tau\bigr)$ \\[0.6em]

% \textbf{Local pairwise behavior}
% & $\operatorname{sign}(\Delta^\star)=\operatorname{sign}(\eta-\tfrac12)$
% & Correct win/loss direction
% & Always decreases disagreement \\[0.6em]

% \textbf{Global population limit}
% & Optimal margin
% $\Delta^\star=\tau\log\frac{\eta}{1-\eta}$
% & $\operatorname{sign}(\Delta)$ win/loss counting
% & $\mathbf{1}[y\Delta<0]$ (Kendall disagreements) \\[0.6em]

% \textbf{Consistency guarantee}
% & Fisher-consistent for pairwise $0$--$1$ loss
% & Converges to Copeland counting
% & Converges to Kemeny objective \\

% \bottomrule
% \end{tabular}
% \label{tab:differential_voting_losses}
% \end{table*}

\begin{figure*}
    \centering
    \includegraphics[width=\linewidth]{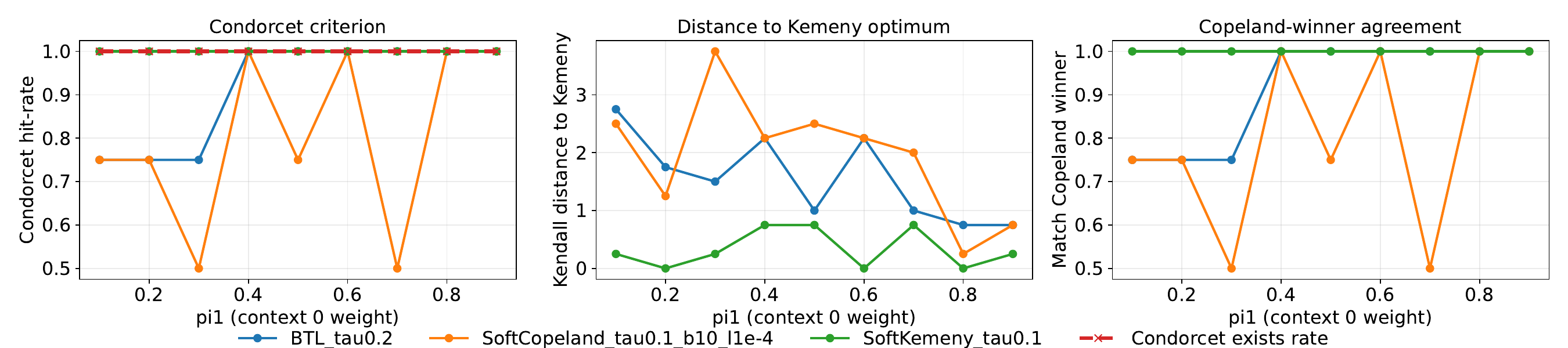}
    \caption{Aggregation under heterogeneous (hidden-context) preferences.
    Left: Condorcet winner criterion; middle: distance to the Kemeny optimum; right: Copeland-winner agreement.
    Soft Copeland and Soft Kemeny inherit the axiomatic behavior of their classical counterparts,
    whereas BTL systematically violates Condorcet consistency as context mixtures vary.}
    \label{fig:exp2}
\end{figure*}

\begin{figure}[t]
    \centering
    \includegraphics[width=\linewidth]{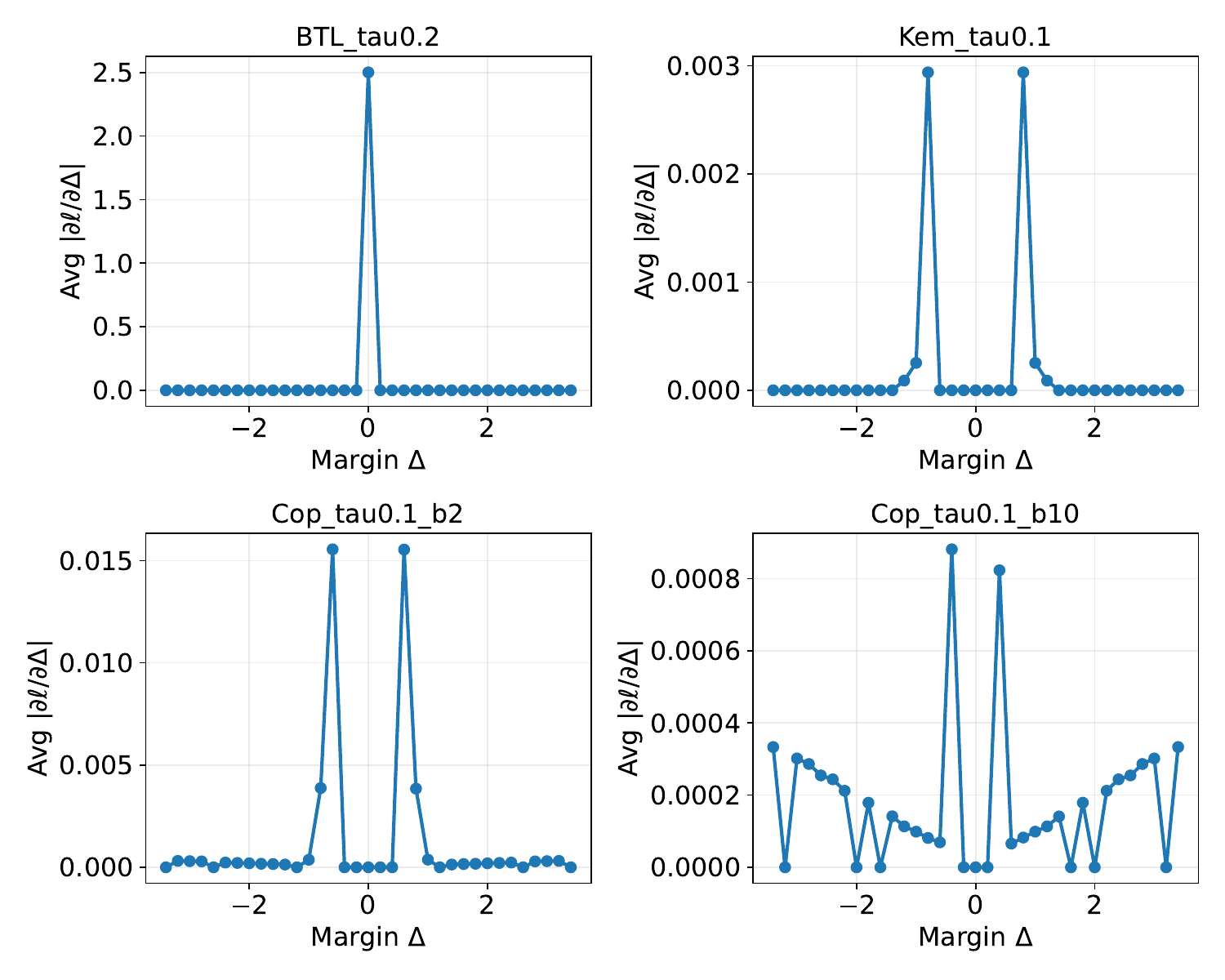}
    \caption{Loss geometry analysis.
    Average gradient magnitude as a function of margin $\Delta$ at convergence.
    BTL concentrates gradients near zero margins, Soft Kemeny targets pairwise disagreement boundaries,
    and Soft Copeland exhibits tunable majority-style saturation controlled by $\beta$.}
    \label{fig:exp3}
\end{figure}

\section{Experiments}

Our theoretical analysis shows that the choice of loss function in preference
optimization implicitly determines a social choice rule at the population
level. We evaluate this claim empirically through a controlled suite of
synthetic experiments designed to isolate aggregation behavior, axiomatic
properties, and optimization geometry.

\subsection{Research Questions}

Our experiments are guided by the following three research questions.

\paragraph{RQ1: Population-level rule recovery.}
Do differentiable losses recover their intended classical aggregation rules
(Copeland or Kemeny) at the population level as smoothing parameters vanish, and
how does this behavior compare to standard BTL-style losses?

\paragraph{RQ2: Axiomatic behavior under heterogeneous preferences.}
In the presence of heterogeneous (hidden-context) preferences, do different
losses satisfy or violate classical axioms such as the Condorcet winner
criterion and consistency with Copeland or Kemeny outcomes?

\paragraph{RQ3: Loss geometry and optimization behavior.}
How do differences in loss geometry (e.g., gradient concentration and saturation)
explain the distinct normative and optimization behaviors of the proposed losses?

\subsection{Experimental Design}

All experiments are implemented in a single, fully reproducible script and use
minimal reward models to isolate aggregation effects.

\paragraph{Synthetic preference generation.}
We consider $m\in\{5,7,9\}$ alternatives and generate pairwise preferences from
known population-level structures:
\emph{(i)} transitive utilities,
\emph{(ii)} near-tie utilities,
\emph{(iii)} cyclic (Condorcet-cycle) preferences, and
\emph{(iv)} sharply transitive utilities.
Pairwise labels are sampled from Bernoulli distributions with probabilities
$\eta_{ab}$ derived from these structures.

We compare the standard BTL (logistic) loss against our proposed Soft Copeland
and Soft Kemeny losses. The reward model is a vector $r\in\mathbb{R}^m$ (one scalar
per alternative), ensuring that differences arise solely from the loss design.

\paragraph{Experiment 1: Population recovery under smoothing.}
We optimize each loss across a grid of temperature parameters $\tau$ (and, for
Soft Copeland, saturation $\beta$ and regularization $\lambda$). We measure:
(i) agreement with the Copeland winner and
(ii) Kendall distance to the Kemeny-optimal ranking.
This experiment directly addresses RQ1.

\paragraph{Experiment 2: Hidden-context aggregation.}
To model heterogeneous preferences, each comparison is generated by first
sampling a latent context, each with its own utility function. We evaluate:
(i) satisfaction of the Condorcet winner criterion,
(ii) distance to the Kemeny optimum, and
(iii) agreement with the Copeland winner as context mixture weights vary.
This experiment addresses RQ2.

\paragraph{Experiment 3: Loss geometry analysis.}
We analyze the gradient magnitude $\lvert \partial \ell / \partial \Delta \rvert$
as a function of the margin $\Delta$ for each loss at convergence. This
experiment isolates how loss geometry concentrates optimization pressure near
decision boundaries and addresses RQ3. (We omit additional noise-robustness
experiments for clarity.)

\subsection{Results}

\paragraph{Experiment 1: Recovery of classical aggregation rules.}
Figure~\ref{fig:exp1} shows population-level recovery across four preference
regimes.
BTL exhibits stable but systematically misaligned behavior: it fails to recover
Copeland winners in cyclic and near-tie settings and maintains large Kendall
distance to Kemeny optima across all regimes.
In contrast, Soft Copeland reliably recovers the Copeland winner when a stable
majority structure exists, while exhibiting sensitivity in near-degenerate
cases, mirroring the behavior of classical Copeland aggregation.
Soft Kemeny consistently achieves substantially lower Kendall distance than BTL.

\paragraph{Experiment 2: Axioms under heterogeneous preferences.}
Figure~\ref{fig:exp2} evaluates aggregation under hidden contexts.
BTL frequently violates the Condorcet winner criterion and diverges from both
Copeland and Kemeny outcomes as context mixtures vary.
Soft Copeland closely tracks the existence of a Condorcet winner and selects it
whenever it is stable, while Soft Kemeny always satisfies the Condorcet
criterion and achieves the smallest distance to the Kemeny optimum.
These results empirically confirm that different differentiable losses implement
distinct axiomatic commitments under aggregation.

\paragraph{Experiment 3: Loss geometry.}
Figure~\ref{fig:exp3} visualizes gradient magnitude as a function of margin.
BTL concentrates gradients sharply at $\Delta\approx 0$ and rapidly decays,
reflecting margin calibration rather than win/loss counting.
Soft Kemeny exhibits symmetric peaks near decision boundaries and vanishing
gradients for confidently ordered pairs, consistent with minimizing pairwise
disagreement.
Soft Copeland shows tunable boundary concentration controlled by $\beta$,
interpolating between smooth optimization and near-discrete majority counting.
These geometric differences explain the normative and optimization behavior
observed in Experiments~1 and~2.

\subsection{Answers to Research Questions}

\paragraph{RQ1.}
Yes. Soft Copeland and Soft Kemeny recover their intended classical aggregation
rules at the population level in the appropriate limits, while BTL converges to
a distinct, scoring-based objective.

\paragraph{RQ2.}
Under heterogeneous preferences, loss choice determines which axioms are
satisfied. Soft Copeland and Soft Kemeny inherit the Condorcet consistency of
their classical counterparts, whereas BTL systematically violates it.

\paragraph{RQ3.}
Differences in loss geometry directly explain differences in aggregation
behavior. Boundary concentration and saturation encode normative assumptions
about whether margins, wins, or disagreements should dominate optimization.

Overall, these experiments demonstrate that preference optimization losses are
not interchangeable statistical surrogates: they implement qualitatively
different social choice rules. Differential Voting makes these choices explicit,
controllable, and compatible with gradient-based learning.

\section{Limitations}

Our analysis focuses on population-level objectives and their minimizers, rather
than finite-sample guarantees or full RLHF optimization dynamics; in practice,
training may be affected by sampling noise, model misspecification, and
interactions with policy optimization that we do not study. We consider
differentiable surrogates for deterministic classical voting rules in controlled
synthetic settings; extending this framework to probabilistic aggregation,
multi-winner objectives, or large-scale real-world annotation pipelines remains
future work. Finally, Differential Voting does not eliminate the trade-offs
imposed by classical impossibility results but makes them explicit at the level
of loss design.

% \section*{Acknowledgments}

%% Bibliography
\bibliographystyle{named}
\bibliography{ijcai26}

\begin{thebibliography}{}

\bibitem[\protect\citeauthoryear{Burges \bgroup \em et al.\egroup }{2005}]{burges2005learning}
Chris Burges, Tal Shaked, Erin Renshaw, Ari Lazier, Matt Deeds, Nicole Hamilton, and Greg Hullender.
\newblock Learning to rank using gradient descent.
\newblock In {\em Proceedings of the 22nd international conference on Machine learning}, pages 89--96, 2005.

\bibitem[\protect\citeauthoryear{Ge \bgroup \em et al.\egroup }{2024}]{ge2024axioms}
Luise Ge, Daniel Halpern, Evi Micha, Ariel~D Procaccia, Itai Shapira, Yevgeniy Vorobeychik, and Junlin Wu.
\newblock Axioms for ai alignment from human feedback.
\newblock {\em Advances in Neural Information Processing Systems}, 37:80439--80465, 2024.

\bibitem[\protect\citeauthoryear{G{\"o}lz \bgroup \em et al.\egroup }{2025}]{golz2025distortion}
Paul G{\"o}lz, Nika Haghtalab, and Kunhe Yang.
\newblock Distortion of ai alignment: Does preference optimization optimize for preferences?
\newblock {\em arXiv preprint arXiv:2505.23749}, 2025.

\bibitem[\protect\citeauthoryear{Maura-Rivero \bgroup \em et al.\egroup }{2025}]{maura2025jackpot}
Roberto-Rafael Maura-Rivero, Marc Lanctot, Francesco Visin, and Kate Larson.
\newblock Jackpot! alignment as a maximal lottery.
\newblock {\em arXiv preprint arXiv:2501.19266}, 2025.

\bibitem[\protect\citeauthoryear{Munos \bgroup \em et al.\egroup }{2024}]{munos2024nash}
R{\'e}mi Munos, Michal Valko, Daniele Calandriello, Mohammad~Gheshlaghi Azar, Mark Rowland, Zhaohan~Daniel Guo, Yunhao Tang, Matthieu Geist, Thomas Mesnard, C{\^o}me Fiegel, et~al.
\newblock Nash learning from human feedback.
\newblock In {\em Forty-first International Conference on Machine Learning}, 2024.

\bibitem[\protect\citeauthoryear{Ouyang \bgroup \em et al.\egroup }{2022}]{ouyang2022training}
Long Ouyang, Jeffrey Wu, Xu~Jiang, Diogo Almeida, Carroll Wainwright, Pamela Mishkin, Chong Zhang, Sandhini Agarwal, Katarina Slama, Alex Ray, et~al.
\newblock Training language models to follow instructions with human feedback.
\newblock {\em Advances in neural information processing systems}, 35:27730--27744, 2022.

\bibitem[\protect\citeauthoryear{Rafailov \bgroup \em et al.\egroup }{2023}]{rafailov2023direct}
Rafael Rafailov, Archit Sharma, Eric Mitchell, Christopher~D Manning, Stefano Ermon, and Chelsea Finn.
\newblock Direct preference optimization: Your language model is secretly a reward model.
\newblock {\em Advances in neural information processing systems}, 36:53728--53741, 2023.

\bibitem[\protect\citeauthoryear{Siththaranjan \bgroup \em et al.\egroup }{2024}]{siththaranjan2023distributional}
Anand Siththaranjan, Cassidy Laidlaw, and Dylan Hadfield-Menell.
\newblock Distributional preference learning: Understanding and accounting for hidden context in rlhf.
\newblock {\em ICLR}, 2024.

\bibitem[\protect\citeauthoryear{Stiennon \bgroup \em et al.\egroup }{2020}]{stiennon2020learning}
Nisan Stiennon, Long Ouyang, Jeffrey Wu, Daniel Ziegler, Ryan Lowe, Chelsea Voss, Alec Radford, Dario Amodei, and Paul~F Christiano.
\newblock Learning to summarize with human feedback.
\newblock {\em Advances in neural information processing systems}, 33:3008--3021, 2020.

\end{thebibliography}

\end{document}